\documentclass[12pt]{article}
\usepackage[T1,T2A]{fontenc}
\usepackage[lutf8]{luainputenc}
\usepackage[english,russian]{babel}
\usepackage{graphics,graphicx}
\usepackage{amssymb,amsfonts,amsthm,amsmath,mathtext,enumerate}
\usepackage{cite}
\usepackage{titlesec}
\usepackage{dsfont}
\usepackage{array}
\usepackage{url}
\usepackage{soul}
\usepackage[usenames]{color} 
\usepackage{algorithm,algorithmicx}
\usepackage{algpseudocode}
\usepackage[caption=false]{subfig}
\usepackage{cases}
\usepackage{multirow}
\usepackage{booktabs}
\usepackage{colortbl}
\usepackage{ragged2e}
\usepackage[margin=2em,labelsep=space,skip=0.5em,font=normalfont]{caption}
\usepackage{nccmath}

\titleformat{\section}[runin]{\normalfont\bfseries}{\indent\thesection}{1em}{}[]
\titleformat{\subsection}[runin]{\normalfont}{\indent\thesubsection}{1em}{\so}[]
\makeatletter\renewcommand{\@biblabel}[1]{#1.}\makeatother
\def\thesection{\arabic{section}.}
\def\thesubsection{\arabic{section}.\arabic{subsection}.}
\numberwithin{equation}{section}

\newfont{\bb}{msbm10}

\def\tf{{\tilde f}}

\def\mathO{\mathit{O}}  

\newcounter{definition}
\newenvironment{definition}[1][]{\refstepcounter{definition}\par\smallskip
	{\so {Определение~\thedefinition. #1}} \rmfamily}{\smallskip}
\newcounter{method}

\newcounter{remark}

\newcounter{theorem}
\newenvironment{theorem}[1][]{\refstepcounter{theorem}\par\smallskip
	{\so {Теорема~\thetheorem. #1}} \rmfamily}{\smallskip}
\newcounter{lemma}
\newenvironment{lemma}[1][]{\refstepcounter{lemma}\par\smallskip
	{\so {Лемма~\thelemma. #1}} \rmfamily}{\smallskip}

\def\RightNote#1{\vadjust{\setbox1=\vtop{\parindent=0pt\hsize 28mm \footnotesize \baselineskip=9pt \leftskip=3mm \rightskip=4mm plus 4mm\hfill#1}\hbox to\hsize{\hfill\rlap{\smash{\raise .5ex\box1}}}}}
\def\LeftNote#1{\vadjust{\setbox1=\vtop{\hsize 50mm\parindent=6pt \baselineskip=9pt \rightskip=4mm plus 4mm{\bf #1}} \hbox{\kern-2.7cm\smash{\raise .5ex\box1}}}}

\begin{document}

\noindent
\MakeUppercase{Компьютерные методы}\\
УДК 
004.021, 
519.677 

\begin{center} 
\textbf{\MakeUppercase{Метод попарного сходства для задачи поиска числа доминирования}}$^1$ \\

\vspace{4mm}
\textbf{\copyright 2025 г.\;\;
Н.И.Шушко$^{a,**}$, Д.В.Лемтюжникова$^{a,*}$,}\\ 
$^a$Институт проблем управления РАН им. В. А. Трапезникова, Москва, Россия \\
\textit{* e-mail: darabbt@gmail.com} \\
\textit{** e-mail: shushko.ni@phystech.edu} \\

\smallskip
\smallskip
Поступила в редакцию  \\
После доработки  \\
Принята к публикации 
\end{center}

\footnotetext{Работа выполнена при финансовой поддержке
 Российского научного фонда (проект 22-71-10131).Статья рекомендована к
публикации программным комитетом конференции ИОИ-2024 (23-27 сентября
2024 г., г. Гродно) по результатам доклада.} 

{\footnotesize
В работе рассматривается задача поиска числа доминирования в процедурах двухуровневого голосования, где на первом этапе голосование проводится в локальных группах агентов, а на втором этапе результаты агрегируются в итоговое решение. Основной целью является определение минимальной доли агентов, поддерживающих предложение, при которой оно может быть принято. В работе используется метод попарного сходства для анализа структуры задачи и построения эвристических алгоритмов с гарантированной точностью. Рассмотрены три специальных случая: граф связи агентов как дерево, полный граф и регулярный граф с нечетным количеством вершин. Для каждого случая предложены новые эвристические алгоритмы и функции попарного сходства, позволяющие оценить погрешность решения. Результаты работы расширяют применение полиномиальных алгоритмов на более широкий класс задач и предоставляют критерии для выбора оптимального алгоритма на этапе постобработки.
}

\textbf{Ключевые слова: }метод попарного сходства, двухуровневое голосование, эвристический алгоритм, регулярный граф, полный граф, граф дерево.
\newpage
\begin{center} 
\textbf{\MakeUppercase{Pairwise similarity method for majority domination problem}} \\

\vspace{4mm}

\textbf{\copyright 2025 y.\;\; N.I. Shushko$^{a,*}$, D.V. Lemtyuzhnikova$^{a,**}$}\\
$^a$ V.A. Trapeznikov Institute of Control Sciences of Russian Academy of Sciences, Moscow, Russia \\
\textit{* e-mail: shushko.ni@phystech.edu} \\
\textit{** e-mail: darabbt@gmail.com} 

\end{center}

{\footnotesize
The paper considers the problem of finding the number of dominant voters in two-level voting procedures. At the first stage, voting is conducted among local groups of voters, and at the second stage, the results are aggregated to form a final decision. The goal is to determine the minimum proportion of voters supporting a proposal for it to be accepted.
The paper uses the method of pairwise comparisons to analyze the structure of the problem and develop heuristic algorithms with guaranteed accuracy. Special cases are considered, including the agent communication graph as a tree, complete graph, or regular graph with an odd number of vertices. New heuristic algorithms are proposed for each case, along with pairwise comparison functions to estimate the accuracy of the solution.
Results extend the use of polynomial algorithms to a broader class of problems, providing criteria for selecting the optimal algorithm during the post-processing stage.
}


\section*{\quad Введение.} \label{intro}
В социальных и технических системах для принятии решений нередко используются процедуры двухуровневого голосования\cite{breer2017}. В таких процедурах на первом этапе голосование проводится в локальных группах агентов, на втором этапе результаты голосования групп агрегируются -- также посредством голосования -- в итоговое решение~\cite{chebotarev2023power}. Один из центральных вопросов при анализе двухуровневых процедур: <<При какой минимальной доле агентов, поддерживающих предложение, оно может быть в итоге принято данной процедурой?>> Для ответа на этот вопрос исследуется разность числа агентов, поддерживающих предложение, и агентов, не поддерживающих предложение Данную величину называют числом доминирования. 

В \cite{BROERE1995125} было показано, что для случая произвольного графа задача поиска числа доминирования является $NP$-полной.

В работе \cite{yeh1997algorithmic} дается оценку сложности задачи поиска числа доминорования $O(n^2)$ для графа связи агентов дерево, $O(n^3)$ для кографов, и полиномиальную сложность для $k$-деревьев для любого фиксированного $k$.

В \cite{holm2001majority} рассматриваются значения числа доминирования для многодольных графов и их объединений. Доказывается, что задача поиска числа доминирования для графа $G = \bigcup_i K_{n_i}$, где $K_{n_i} -$ это полный граф размером $n_i$, является $NP$-полной. Однако, сложность задачи для многодольного графа $G = K_{n_1, \ldots, n_m}$ в работе не дается ($3$-й пункт в списке открытых проблем).

В работе \cite{chebotarev2023power} была рассмотрена задача двухуровневого голосования на регулярном графе, даны аналитические формулы значений числа доминирования для каждого регулярного графа с нечетным количеством вершин. 

В \cite{lemtyuzhnikova2023pairwise} предлагается эвристический алгоритм решения задачи поиска числа доминирования c гарантированной абсолютной погрешностью основанный на методе попарного сходства.

В данной работе предлагается уточнить лемму об удалении ребра из графа, которая была сформулирована в работе \cite{lemtyuzhnikova2023pairwise}, и предложить модификацию алгоритма с более точной функцией попарного сходства, что позволяет сделать более точную оценку абсолютной погрешности алгоритма. Также опишем процесс построения эвристических алгоритмов на основе регулярных графов. Предворительные результаты исследования были анонсированы в докладе \cite{Sh2024}.

\section{Постановка задачи.}
Рассматривается граф $G = (V, E)$, $|V| = n$, каждая вершина которого имеет петлю. Вершина $v \in V$ имеет множество соседей $N_v = \{w \in V : (v, w) \in E\}$, в которое входит и~$v$.  На множестве вершин $V$ задается функция мнений $f: V\to\{-1, 1\}$ , которая каждой вершине ставит в соответствие ее мнение. Функция мнений расширяется на подмножества $W \subseteq V$ как сумма мнений вершин множества $W$ по формуле $f(W) = \sum_{v \in W} f(v)$. В ходе голосования вершина $v \in V$ учитывает мнение вершин
из $N_v$, и {\em голосует <<за>>}, если $f(N_v)>0,$ либо {\em<<против>>} в ином случае. Пару $(G,f)$ будем называть {\em конфигурацией}. Пусть |$V^+|= \{v \in V : f(N_v) >0\}$ -- множество голосующих <<за>> вершин.

Предложение, поставленное на голосование, принимается, если число голосующих <<за>> вершин строго больше $|V|/2$. Функцию мнений на которой предложение принимается будем называть строго мажоритарной.
Под результатом голосования будем понимать принятие либо отклонение предложения.

\begin{definition} 
Числом строгого доминирования для графа $G$ называется величина $\gamma(G) = \min\limits_f\{f(V): |V^+| > |V|/2\}$, где $f(V)$ -- строго мажоритарная функция мнений на множестве~$V$. 
Функцию мнений графа $G$ будем называть {\em оптимальной\/}, если на ней достигается число строгого доминирования $\gamma(G)$.
\end{definition}

Рассматривается задача поиска такой функции мнений $f(v)$ графа $G$, которая реализует число строгого доминирования $\gamma(G)$, то есть
обеспечивает успех голосования при минимальном числе вершин $v$ таких, что $f(v) = 1$.

\section{Метод попарного сходства.}
В данной работе предлагается использовать метод попарного сходства  для расширения применения алгоритмов для простых специальных случаев и с помощью функции попарного сходства анализировать успешность такого расширения. 

 Метод попарного сходства обобщает метрический подход, который изначально был сформулирован для задач теории расписаний\cite{lazarev2021metric,  lazarev2020polynomially, lazarev2021metricB, bukueva2022analysis, cheng2022metric}. В нем учитываются разные структуры примеров задач и для сравнения задач используется функция попарного сходства, для которой соблюдение аксиом метрики не является обязательным. 
 Под примером задачи будем понимать конкретный фиксированный набор параметров задачи, в который входит количество переменных, коэффициенты целевой функции и матрица ограничений.
\begin{definition}  
 {\em Специальным случаем} $P^{\alpha}$ будем называть подкласс исходной 
  $NP-$полной задачи $P$, где для всех примеров задач $A \in P^{\alpha}$ выполняется закономерность~$\alpha$. 
\end{definition}
\begin{definition}
 {\em Функцией попарного сходства} будем называть функцию, определяющую различие примеров по некоторому критерию.
\end{definition}  
 
Пусть дан простой частный случай $P^{\alpha}$ и соответствующий полиномиальный или псевдополиномиальный алгоритм. Тогда для любых $B \in P^{\alpha}$ и $A\not\in P^{\alpha}, $ значение функции различия $\tau$ на паре $(A,B)$ равна абсолютной (или относительной) погрешности целевой функции $A$ при использовании в качестве решения полиномиального (псевдополиномиального) точного решения задачи ~$B$.

Также могут быть использованы как критерии сходимость алгоритмов и время работы алгоритмов. В данной работе используется критерий на основе точности. 

В работе \cite{lemtyuzhnikova2023pairwise} предлагается следующие шаги построения и использования функции попарного сходства:

\begin{enumerate}
    \item Найти один или несколько специальных случаев задачи. 
    
    \item Построить функцию попарного сходства, которая позволяет оценивать расстояние от любой задачи до каждого из возможных специальных случаев.
    
    \item Найти ближайший специальный случай к заданному примеру $A$, используя выбранную функцию попарного сходства.
    \item \label{i:Find}
    Найти решение наиболее близкого примера к $A$ из множества примеров, которые являются специальным случаем.
      
    \item Оценить гарантированную погрешность полученного решения с помощью соответствующей функции попарного сходства.
    
    \item Проверить полученное решение на допустимость, и если какие-то ограничения нарушены, то использовать функцию, преобразующую $X_B$ в $X_A=X^*_B.$
\end{enumerate}

Рекомендации, по нахождению новых специальных случаев, а также алгоритм для нахождения функции попарного сходства можно найти в \cite{lemtyuzhnikova2023pairwise}.

После построения функции попарного сходства предлагается синтезировать алгоритм поиска допустимого решения и получить оценку абсолютного отклонения данного решения от оптимального. 

На основе функции различия для нескольких специальных случаев, можно принять решение о выборе эвристического алгоритма для решения исходной задачи.

В следующих главах будут рассмотрены три специальных случая: граф связи как дерево, полных граф и регулярный граф для нечетного количества вершин. Для каждого случая будет предложен новый эвристический алгоритм, а также алгоритм выбора одной из этих эвристик для получения решения исходной задачи.

\section{Специальный случай на основе графа дерева.}
 В работе \cite{lemtyuzhnikova2023pairwise} представлен алгоритм для решения задачи поиска числа строгого доминирования. Он заключается в переходе от исходной задачи к её специальному случаю — задаче с графом мнений в виде остовного дерева исходного графа. Затем алгоритм возвращает решение задачи на дереве к решению исходной задачи путём добавления рёбер и изменения функции мнений, с условием сохранения результата голосования. Таким образом алгоритм можно представить как переход между $k$ графами, где $k$ - цикломатическое число исходного графа. 
 На основе леммы 1 индукционно строилась оценка точности предлагаемого алгоритма, где отношение из леммы 1 дает оценку точности на одном переходе между графами.
\begin{lemma}
 Пусть $G$ -- непустой граф, а $H$ -- граф, получаемый из $G$ удалением произвольного ребра $uv$. Тогда
\begin{equation}\label{eq_cycles}
-4\, \le\, \gamma(G) - \gamma(H)\, \le\,  2.
\end{equation}
\end{lemma}
Данное отношение справедливо для любых двух графов отличающихся на одно ребро. В работе предлагается уточнить лемму 1 для случаев удаления и добавления ребра для трех комбинаций мнений смежных вершин между которыми добавляется или удаляется ребро: <<за-за>>, <<против-против>>, <<за-против>>. Полученный результат позволит построить более точную оценку отклонения полученного решения от оптимального, а также позволит построить алгоритм выбора остовного дерева.

Расширим ее формулировку рассмотрев отдельно два случая: случай когда известна оптимальная функция мнений графа $G$ и случай когда известна оптимальная функция мнений графа $H$ соответственно. 

\begin{definition}
 Пусть заданы конфигурации $G^f$ и $G'^{\tf}$ с общим множеством вершин~$V$. Будем говорить, что {\em условия голосования\/} вершины $v$ в графе $G'^{\tf}$ лучше по сравнению с графом $G^f$, если  $v$ имеет не меньше соседей в $G'^{\tf}$ с мнением $+1$ чем в $G^f$ и не больше соседей с мнением $-1$ в $G'^{\tf}$ чем в~$G^f$.
\end{definition}

\begin{lemma}    
 Пусть $G$ -- непустой граф, а $H$ -- граф, получаемый из $G$ удалением произвольного ребра $uv$. Пусть мы знаем оптимальную функцию мнений графа $G$. Тогда если
\begin{itemize}
    \item $f(u)=f(v)=1$, то 
    \begin{equation}
-4\, \le\, \gamma(G) - \gamma(H)\, \le\,  2;
\end{equation}
    \item $f(u)=f(v)=-1$, то 
    \begin{equation}
0\, \le\, \gamma(G) - \gamma(H)\, \le\,  2;
\end{equation}
    \item в остальных случаях 
    \begin{equation}
-2\, \le\, \gamma(G) - \gamma(H)\, \le\,  2
\end{equation}
\end{itemize}
\end{lemma}
\begin{proof}
Для построения оценки отклонения числа строгого доминирования полученного решения, оценим  $\gamma(H)$ сверху и снизу.

Так как добавление или удаление ребра между двумя вершинами с мнением <<против>> не изменит результат голосования этих вершин, нас интересуют случаи, когда вершины $u$ и/или $v$ в графе $H^f$ голосуют <<за>>. В таком случае эти голоса являются решающими для принятия предложения, без них предложение будет отвергнуто. В этом случае после добавления ребра условия голосования этих вершин могут ухудшиться, что может привести к невыполнению условия строгой мажоритарности функции мнений на новом графе $G^f$. В случае, когда вершины голосуют против, изменения условий голосования этих вершин не отменяют принятия предложения и сохраняют строгую мажоритарность функции~$f$ в $G^f$.

Пусть существует оптимальная функция мнений $f$ на графе $H$, тогда $f(V) = \gamma(H)$. Оценим число строгого доминирования на графе $G$. Для этого рассмотрим случаи, когда функция мнений $f$ после добавления ребра $uv$ к графу $H$ перестает быть строго мажоритарной функцией мнений на новом графе $G$, и минимально изменим $f$ для получения новой функции $\tilde f$, которая будет строго мажоритарной. Значение  $\tilde f(V)$ служит верхней оценкой $\gamma(G)$.

Рассмотрим все комбинации мнений вершин $u$ и $v$:

\begin{enumerate}
         \item $f(u)=f(v)=1.$ В этом случае добавление соседа с мнением $+1$, не ухудшит условий голосования для обеих вершин. Следовательно функция мнений $f$ остается строго мажоритарной на графе $G$, но может перестать быть оптимальной. Отсюда $\gamma(G) \le \tf(V)=f(V)=\gamma(H).$
        
    \item $f(u)=f(v)=-1.$ Добавление связи между $u$ и $v$ изменит условия голосования в худшую сторону для обеих вершин.
        \begin{itemize}
            \item Если $u$ и $v$ голосовали <<за>> в $H^f$ и $f(N_u(H))=1$ или $f(N_v(H))=1$, после добавления ребра $uv$ результат голосования одной из вершин изменится на <<против>> и функция мнений перестанет быть строго мажоритарной. Для перехода к строго мажоритарной функции мнений достаточно поменять мнение одной из вершин с $-1$ на $+1$. Тогда у этой вершины  $\sum \tf(k)$ по соседям останется таким же, как и при конфигурации $H^f$, а у второй вершины увеличится на $1$. Следовательно голоса сохранят свое  и функция $\tf$ будет строго мажоритарной. В этом случае $\tf(V)=f(V)+2=\gamma(H)+2$, так как изменилось только одно мнение с $-1$ на $+1$.
            \item Если одна из вершин голосовала <<за>>, а вторая <<против>>, то для получения строго мажоритарной функции мнений $\tf$ достаточно изменить мнение одной из вершин $u$ или $v$ с $-1$ на $+1$. В этом случае также $\tf(V)=f(V)+2=\gamma(H)+2$.
        \end{itemize}
        
    \item $f(u)=-1$, $f(v)=1.$ Для вершины $u$ добавление соседа $v$ не ухудшит условий голосования. Если $v$ голосует против в $H^f$, то добавление соседа $u$ не нарушит строгой мажоритарности функции $f$, так как не приведет к потере голоса <<за>>. В случае голосования $v$ <<за>> для получения строго мажоритарной функции мнений $\tf$ достаточно изменить мнение $u$ на $+1$. Изменение мнения одной вершины с $-1$ на $+1$ увеличит значение функции $\tf(V)$ относительно $f(V)$ на~$2$.
\end{enumerate}

Так как функция мнения для $H$ неизвестна, для оценки $\gamma(H)$ воспользуемся самым худшим случаем:
\begin{equation}\label{equp}
    \gamma(G) - \gamma(H)\, \le\,  2. 
\end{equation}

 Пусть теперь мы имеем оптимальную функцию мнений $f$ на графе $G$. Аналогично оценим сверху $\gamma(H)$, задав строго мажоритарную функцию мнений $\tf$, получаемую из $f$ минимальными изменениями:

\begin{enumerate}
    \item $f(u)=f(v)=-1.$ Удаление связи с вершиной с мнением $-1$ не ухудшает условий голосования для обеих вершин $u$ и $v$. Тогда функция мнений $f$ остается строго мажоритарной для графа $H$ и $\tf=f,$ то есть $\tf(V)=f(V)= \gamma(G)$.
    
    \item $f(u)=f(v)=1.$ Удаление ребра $uv$ приводит к ухудшению условий голосования для обеих вершин $u$ и $v$.

    \begin{itemize}
        \item Если $u$ или $v$ голосует <<за>> в $G^f$, то функция мнений $f$ может перестать быть строго мажоритарной на графе $H$. При голосовании <<за>> вершины $u$ это происходит, если $f(N_u(G))=1$. После удаления ребра $uv$ эта сумма для функции мнений $f$ равна нулю на графе~$H$. Для получения строго мажоритарной функции мнений изменим мнение одной из соседних вершин вершины $u$ с $-1$ на $1$.
        
        Такую замену можно сделать всегда, так как невозможно получить $f(N_u(H))=0$, без хотя бы одной соседней вершины с мнением $-1$, поскольку мнение вершины $u$ равно~$1$. После замены сумма мнений увеличится на $2$ и вершина $u$ будет голосовать <<за>>. В этом случае значение функции мнений $\tf(V)=f(V)+2$.
        
        \item Если $u$ и $v$ голосуют <<за>> в $G^f$, то аналогично предыдущему случаю $f$ может перестать быть строго мажоритарной на графе $H$ из-за изменения голосов у обеих вершин. Тогда требуется заменить мнения у одной соседней вершины для каждой вершины $u$ и $v$, как в предыдущем случае. Тогда для новой функции мнений $\tf$ имеем $\tf(V)-f(V)=4$, так как потребовалось два изменения мнений с $-1$ на $+1$. Отсюда $\gamma(H) \le \tf(V)=\gamma(G)+4$.
        
    \end{itemize}
    \item $f(u)=1, f(v)=-1$.
    Удаление ребра между вершинами улучшает условия голосования для вершины $u$ и ухудшает для вершины $v$. Если вершина $v$ голосует <<за>> в $G^f$,
    удаление ребра может изменить сумму мнений по соседям в $H^f$ и она перестанет быть строго мажоритарной. Для получения строго мажоритарной функции мнений $\tf$ достаточно изменить мнение вершины $v$ на $\tf(v)=+1$. Сумма мнений соседей вершины $v$ увеличится на $2$, что компенсирует удаление соседа с мнением $+1$, и голос <<за>> сохранится. Тогда значение функции мнений $\tf(V)$ из-за изменения мнения с $-1$ на $+1$ превосходит $f(V)$ на $2$. Поэтому $\gamma(H) \le \tf(V)=\gamma(G) + 2.$
\end{enumerate}

Так как функция $f$ в этом случае известна, оценка $\gamma(G) - \gamma(H)$ может быть записана для каждой комбинации мнений вершин $u$ и $v$. Объединив полученные оценки для каждого случая с оценкой сверху (5) получим утверждение леммы.
\end{proof}
\begin{lemma}
 Пусть $G$ -- непустой граф, а $H$ -- граф, получаемый из $G$ удалением произвольного ребра $uv$. Пусть мы знаем оптимальную функцию мнений графа $H$. Тогда если
\begin{itemize}
    \item $f(u)=f(v)=1$, то 
    \begin{equation}
-4\, \le\, \gamma(G) - \gamma(H)\, \le\,  0;
\end{equation}
    \item $f(u)=f(v)=-1$, то 
    \begin{equation}
-4\, \le\, \gamma(G) - \gamma(H)\, \le\,  2;
\end{equation}
    \item в остальных случаях 
    \begin{equation}
-4\, \le\, \gamma(G) - \gamma(H)\, \le\,  2.
\end{equation}
\end{itemize}
\end{lemma}
\begin{proof}
    Аналогично доказательству леммы 2 рассмотрим две оценки $\gamma(G)$ сверху и снизу. Имея оптимальную функцию мнений графа $H$ мы можем рассмотреть случаи оценки сверху для каждого случая распределения мнений.
    
\end{proof}
\begin{theorem}  
Пусть связный граф $G$ имеет цикломатическое число $k$, а $T$ -- его произвольное остовное дерево.
Тогда, если для получения дерева $T$ разрываются $s$ ребер между вершинами <<за>> и $l$ между остальными $l+s=k$, верно неравенство:
\begin{equation}
\gamma(T)-4k \le \gamma(G) \le \gamma(T) + 2l.
\end{equation}\label{Th1}
\end{theorem}
\begin{proof}
    По определению цикломатического числа, дерево $T$ получается из графа $G$ удалением $k$ ребер. Добавлением этих ребер одного за другим порождает последовательность графов $T\!:=\! H_0, H_1, ..., H_k\!=:\!G$.
    Доказательство теоремы сводится к применению леммы 3 к парам графов $H_{i-1},H_i$, $i=1,...,k$ и сложению между собой получающихся неравенств для этих пар графов.
    
    Приведем пример для двух первых шагов. Пусть на первом шаге ребро добавляются между вершинами <<за>>, а на втором между вершинами <<за>> и <<против>>. 
    Тогда справедливо отношение для первого шага:
    \begin{equation}
        \gamma(T) - 4 \le \gamma(H_0) \le \gamma(T)
    \end{equation}
    На втором шаге выполняется отношение:
    \begin{equation}
        \gamma(H_0) - 4 \le \gamma(H_1) \le \gamma(H_0) +2
    \end{equation}
    После сложения этих неравенств и вычитания $\gamma(H_0)$ получим:
    \begin{equation}
        \gamma(T) - 8 \le \gamma(H_1) \le \gamma(T) +2
    \end{equation}
    Аналогично можно продолжить шаги и после $k$ итераций получить конечное отношение (3.9).
\end{proof}
Теорема 1 задает функцию попарного сходства, согласно которой каждое остовное дерево графа $G$ с цикломатическим числом $k$ находится на расстоянии $k$ от графа~$G$. 

Остовное дерево графа $G$ может быть найдено с помощью алгоритмов поиска в глубину и ширину, которые имеют линейную временную сложность. С другой стороны, остовные деревья относятся к простому специальному случаю задачи вычисления числа доминирования, что позволяет вычислить число доминирования для любого из них за полиномиальное время $\mathO(n^2)$  и использовать оптимальную функцию мнений дерева для построения строго мажоритарной функции мнений графа $G$ , дающей приближенное решение исходной задачи.
Предлагается следующий алгоритм нахождения мажоритарной функции мнений графа $G$ с цикломатическим числом $k$:

\begin{enumerate}
    \item Найти остовное дерево $T$ графа $G$.  Вычислить оптимальную строго мажоритарную функцию мнений для остовного дерева с помощью полиномиального алгоритма.
    \item Последовательно добавить $k$ ребер к остовному дереву $T$ до достижения графа $G$:
        \begin{itemize}
            \item Добавить ребро, которое принадлежит $E(G)\!\smallsetminus\! E(T)$. Проверить выполнимость $|V^+| > |V|/2$ для текущей функции мнений. Если неравенство выполняется, перейти к добавлению следующего ребра.
            \item Определить инцидентную к новому ребру вершину с мнением $-1$. Изменить ее мнение на $+1$.
        \end{itemize}
\end{enumerate}

В соответствии с общей схемой применения метода попарного сходства в качестве приближенного числа доминирования графа $G$ необходимо выбрать наилучшее решение, которое удается построить на основе оптимальных функций мнений остовных деревьев $G$.

Неравенство \ref{Th1} задает оценку точности полученного решения. Вычислительная сложность алгоритма равняется сложности поиска числа строгого доминирования на дереве  $\mathO(n^2)$.

\section{Специальный случай на основе полного графа.}

Рассмотрим специальный случай задачи нахождения числа доминирования с графом связей в виде полного графа. В этом случае минимальное число доминирования равно двум в случае четного количества вершин и 1 в случае нечетного количества:
\begin{equation}
    \gamma(G) = 
    \begin{cases}
    1, n - \text{нечетное};\\
    2, n - \text{четное};
    \end{cases}
\end{equation}
Данное число реализуется на функции мнений с $(n-1) //  2$ количеством вершин с мнением <<против>> и остальными вершинами <<за>>.

Для получения полного графа необходимо добавить к исходному графу $\frac{n(n-1)}{2} - p$ ребер, где $p = |E|$  - количество ребер исходного графа. Построим на основе леммы 2 функцию попарного сходства.

\begin{theorem}    
 Пусть $G$ -- непустой граф, а $H$ -- граф, получаемый из $G$ удалением $k$ ребер, где $k = l+s+m$, $l$ - количество удаленных ребер смежных вершинам с мнением <<за>>, $s$ - количество удаленных ребер смежных вершинам с мнением <<против>>, $m$ - количество остальных ребер. Тогда 
\begin{equation}
-4l-2m\, \le\, \gamma(G) - \gamma(H)\, \le\, 2k.
\end{equation} \label{Th2}
\end{theorem}
\begin{proof}
Удаление $k$ ребер одного за другим порождает последовательность графов $G\!:=\! H_0, H_1, ..., H_{k}\!=:\!H$. Доказательство теоремы сводится к применению  леммы 2 к парам графов $H_{i-1},H_i$, $i=1,...,l+m$ и сложению между собой получающихся неравенств \eqref{eq_cycles} для этих пар графов.
\end{proof}

Теорема 2 задает функцию попарного сходства для специального случая полного графа, согласно которой расстояние от полного графа до исходного $G$ равно количеству ребер, необходимое для дополнения графа $G$ до полного. Так как в полном графе все вершины связаны со всеми, оптимальная раскраска достигается на любой функции мнений с необходимым количеством вершин с мнением <<за>>. Поэтому возможно минимизировать отклонение числа доминирования от оптимального для исходного графа за счет удаления ребер между вершинами с мнениями <<против>>. Для этого необходимо вершинам между которыми удаляются вершины в первую очередь задать определить мнение <<против>>.

Для определения мажоритарной функции мнений графа на основе специального случая полного графа предлагается следующий алгоритм:
\begin{enumerate}
    \item Достроить граф $G$ до полного $K$. Построить мажоритарную функций мнений вершин с $(n-1) //  2$ количеством вершин с мнением <<против>>. В первую очередь определить мнения <<против>> вершин из множества вершин которые смежны достраиваемым ребрам.
    \item Последовательно удалить $k$ ребер до достижения графа $G$:
        \begin{itemize}
            \item Удалить ребро, которое принадлежит $E(K)\!\smallsetminus\! E(G)$. Проверить выполнимость $|V^+| > |V|/2$ для текущей функции мнений. Если неравенство выполняется, перейти к удалению следующего ребра.
            \item Определить инцидентную к удаленному ребру вершину с мнением $-1$. Изменить ее мнение на $+1$.
        \end{itemize}
\end{enumerate}

Вычислительная сложность полученного алгоритма равна вычислительной сложности обхода графа  $\mathO(|V|+|E|)$. Оценку точности полученного решения можно получить с помощью неравенства \ref{Th2}. Данную оценку можно использовать как критерий выбора алгоритма для решения данного примера задачи.

\section{Специальный случай на основе регулярного графа.}

Рассмотрим специальный случай связного регулярного графа с нечетным количеством вершин\cite{chebotarev2023power}.

Для перехода от исходного графа к регулярному необходимо удалять и добавлять ребра. На основе леммы 1 сформулируем оценку отклонения числа доминирования после удаления $l$ ребер и добавления $m$ ребер:

\begin{theorem}   
 Пусть $G$ -- непустой граф, а $H$ -- граф, получаемый из $G$ удалением $l$ ребер и добавлением $m$ ребер. Тогда 
\begin{equation}
-4l-2m\, \le\, \gamma(G) - \gamma(H)\, \le\,  2l+4m.
\end{equation} \label{Th3}
\end{theorem}
\begin{proof}
Удаление $l$ ребер и добавление $m$ ребер одного за другим порождает последовательность графов $G\!:=\! H_0, H_1, ..., H_l, H_{l+1}, ..., H_{l+m}\!=:\!H$. Доказательство теоремы сводится к применению  леммы 1 к парам графов $H_{i-1},H_i$, $i=1,...,l+m$ и сложению между собой получающихся неравенств \eqref{eq_cycles} соответствующих этим парам графов.
\end{proof}

Для поиска ближайшего специального случая для исходного графа с нечетным количеством вершин предлагается следующий алгоритм:
\begin{enumerate}
    \item Определим степени вершин исходного графа.
    \item Выберем степень регулярности $k$, равную ближайшему целому среднему 
    \footnote{Можно использовать медианное значение.} исходного графа.
    Степень должна быть четной, так как сумма всех степеней в регулярном графе равная $nk$ должна быть чётной.
    \item Генерируем $k-$регулярный граф с тем же количеством вершин.
    \item Сравниваем графы по редактирующему расстоянию $d$ и выбираем с наименьшим значением.
\end{enumerate}

На рис. \ref{fig:reg} представлен результат работы алгоритма. Для графа $G$ найден ближайший $4$-регулярный граф $H$ с редактирующим расстоянием $3$. Сложность представленного алгоритма за исключением поиска редактирующего расстояния составляет $\mathO(|V|+|E|)$. Сложность поиска редактирующего расстояния является экспоненциальной, поэтому предлагается использовать жадный алгоритм поиска на основе сравнения степеней вершин. \LeftNote{Рис.\ref{fig:reg}}
Главная идея заключается в сопоставлении вершин с похожими степенями, а затем минимизации количества изменений рёбер для приведения исходного графа к регулярному до полного совпадения графов. Временная сложность данной эвристики $\mathO(n^2)$.

С помощью предложенного алгоритма найдем ближайший регулярный граф и  аналитическое решение для него:
\begin{equation}
    \gamma(n,k) = \frac{1}{2}(k+1)\max\{\frac{n+1}{2d},1\}+\mathds{1}_{\{5,3\}\{9,5\}}(n,k);
\end{equation}
где $n$ - количество вершин в графе, $k$ - степень регулярного графа, индикатор $\mathds{1}_{\{5,3\}\{9,5\}}(n,k)$ равен $1$ при равенстве $n,k=5,3$ и $n,k=9,5$ соответственно.
Для получения решения исходной задачи добавим и удалим необходимое количество ребер в регулярном графе, с использованием следующего правила:

\begin{itemize}
    \item Добавить/удалить ребро, которое необходимо для получения исходного графа. Проверить выполнимость $|V^+| > |V|/2$ для текущей функции мнений. Если неравенство выполняется, перейти к добавлению/удалению следующего ребра.
    \item Определить инцидентную к новому/удаленному ребру вершину с мнением $-1$. Изменить ее мнение на $+1$.
\end{itemize}

Вычислительная сложность алгоритма определяется вычислительной сложностью генерации $k-$регулярного графа на $n$ вершинах $\mathO(nk^2)$.На основе теоремы 3 и редактирующего расстояния $d=l+m$ получим оценку точности  решения.

\section{Заключение.}

Основной целью применения метода попарного сходства является исследование структур $NP-$полных задач и расширение применимости полиномиальных алгоритмов на более широком классе задач. В данной работе продемонстрировано использование метода попарного сходства для задачи поиска числа доминирования в результате которого обе цели были выполнены.

 Были исследованы специальные случаи задачи поиска числа доминирования для двухэтапного голосования. Рассмотрены три специальных случая: граф связи вершин дерево, полный граф, регулярный граф. Построены критерии расстояния примеров задач до специальных случаев. Метод попарного сходства позволил нам расширить применение эвристических решений для полного графа и регулярного графа и полиномиального алгоритма для дерева на более широкий класс задач.
Для каждого случая построены функции попарного сходства, что позволяет до использования алгоритма получить оценку погрешности решения. Это дает возможность производить выбор алгоритма на этапе постобработки исходной задачи.

В ходе проведения исследований остались открытыми вопросы поиска ближайшего регулярного графа: выбор эвристического алгоритма для поиска редактирующего расстояния и выбора степени вершины регулярного графа. Планируется провести исследования влияния выбора алгоритма поиска редактирующего расстояния и степени вершин регулярного графа на результат алгоритма. 

\pagebreak

\renewcommand{\section}[2]{} 
\begin{center}
\textbf{СПИСОК ЛИТЕРАТУРЫ}
\end{center}

\pagebreak

Подписи к рисункам статьи Шушко и др., ТиСУ 2025, \textnumero 

Рис. \ref{fig:reg}.
{Исходный граф $G$ и $4$-регулярный граф $H$, с редактирующим расстоянием $3$.
    }
\pagebreak
\begin{figure}[ht!]
    \centering
    $G$\hspace*{19em}$H$
    
    \includegraphics[width = 0.9\textwidth]{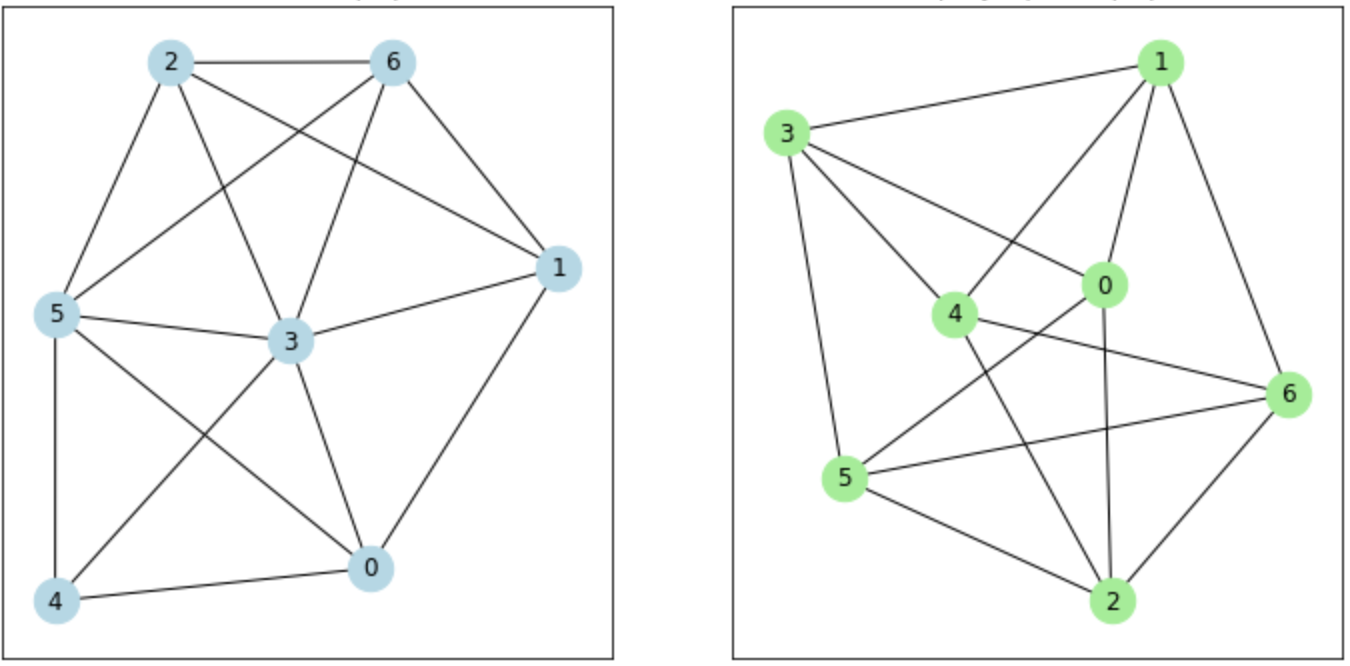}
    \caption{Шушко и др., ТиСУ 2025, \textnumero }
   \label{fig:reg}
\end{figure}
\end{document}